\documentclass[10pt,twocolumn]{IEEEtran}
\usepackage{epsf,psfrag,amssymb,amsfonts,amsmath,cite}
\usepackage{graphicx,subfigure,graphics,color}

\newtheorem{theorem}{Theorem}
\newtheorem{example}{Example}
\newtheorem{corollary}{Corollary}
\newtheorem{lemma}{Lemma}
\newtheorem{definition}{Definition}

\newtheorem{proposition}{Proposition}

\def\psfancypar#1#2{\begingroup\def\par{\endgraf\endgroup\lineskiplimit=0pt}
               \setbox2=\hbox{\large\sc #2}
               \newdimen\tmpht \tmpht \ht2 \advance\tmpht by \baselineskip
               \font\hhuge=Times-Bold at \tmpht
               \setbox1=\hbox{{\hhuge #1}}
               \count7=\tmpht \count8=\ht1
               \divide\count8 by 1000 \divide\count7 by \count8
               \tmpht=.001\tmpht\multiply\tmpht by \count7
               \font\hhuge=Times-Bold at \tmpht
               \setbox1=\hbox{{\hhuge #1}}
               \noindent
                \hangindent1.05\wd1
               \hangafter=-2 {\hskip-\hangindent
               \lower1\ht1\hbox{\raise1.0\ht2\copy1}%
                \kern-0\wd1}\copy2\lineskiplimit=-1000pt}

\newcommand{\beq}{\begin{equation}}
\newcommand{\eeq}{\end{equation}}
\newcommand{\bqa}{\begin{eqnarray}}
\newcommand{\eqa}{\end{eqnarray}}
\newcommand{\bqn}{\begin{eqnarray*}}
\newcommand{\eqn}{\end{eqnarray*}}
\newcommand{\nn}{\nonumber}

\newcommand{\be}{\begin{enumerate}}
\newcommand{\ee}{\end{enumerate}}
\newcommand{\bi}{\begin{itemize}}
\newcommand{\ei}{\end{itemize}}
\newcommand{\bd}{\begin{description}}
\newcommand{\ed}{\end{description}}
\newcommand{\ba}{\begin{array}}
\newcommand{\ea}{\end{array}}
\newcommand{\bde}{\begin{definition}}
\newcommand{\ede}{\end{definition}}
\newcommand{\bex}{\begin{example}}
\newcommand{\eex}{\end{example}}

\def\boxit#1{\vbox{\hrule\hbox{\vrule\kern3pt
        \vbox{\kern3pt#1\kern3pt}\kern3pt\vrule}\hrule}}

\def\reals{ { {\rm  I \kern-0.15em R }  } }
\def\complex{ {\,{{\rm C} \kern-0.50em \raise0.20ex {  |}}\, }}

\def\0bf{{\bf 0}}
\def\1bf{{\bf 1}}
\def\2bf{{\bf 2}}
\def\3bf{{\bf 3}}
\def\4bf{{\bf 4}}
\def\5bf{{\bf 5}}
\def\6bf{{\bf 6}}
\def\7bf{{\bf 7}}
\def\8bf{{\bf 8}}
\def\9bf{{\bf 9}}

\def\xbf{{\bf x}}
\def\ybf{{\bf y}}

\def\xbf{{\bf x}}
\def\ybf{{\bf y}}

\def\Rbf{{\bf R}}

\def\Wbf{{\bf W}}
\def\Xbf{{\bf X}}

\def\Rxx{\Rbf_{\ssstyle X\kern-.1em X}}

\let\ssstyle=\scriptscriptstyle

\def\Kout{\setbox1=\hbox{\Huge\bf K}\hbox to
1.05\wd1{\hspace{.05\wd1}
}}
\def\Sout{\setbox1=\hbox{\Huge\bf S}\hbox to 1.05\wd1{\hspace{.05\wd1}
}}

\def\scalefig#1{\epsfxsize #1\textwidth}
\begin{document}
\title{Decentralized Data Reduction with Quantization Constraints}
%\IEEEoverridecommandlockouts
\author{Ge Xu, Shengyu Zhu, \IEEEmembership{Student Member, IEEE}, and Biao Chen,\IEEEmembership{ Senior Member, IEEE}\thanks{G. Xu, S. Zhu and B. Chen are with the Department of Electrical Engineering and Computer Science, Syracuse University, Syracuse, NY, 13244. Email: gexu\{szhu05, bichen\}@syr.edu. }
\thanks{
%This material is based upon research supported by National Science Foundation under Award CCF1218289, by Army Research Office under Award W911NF-12-1-0383, and by Air Force Office of Scientific Research under Award FA9550-10-1-0458. 
The material in this paper was presented in part at the IEEE International Symposium on Information Theory, Boston, MA, July 2012 and the Asilomar Conference on Signals, Systems, and Computers, Monterey, CA, November 2013.}}

\maketitle

\begin{abstract}
A guiding principle for data reduction in statistical inference is the sufficiency principle. This paper extends the classical sufficiency principle to decentralized inference, i.e., data reduction needs to be achieved in a decentralized manner. We examine the notions of local and global sufficient statistics and the relationship between the two for decentralized inference under different observation models. We then consider the impact of quantization on decentralized data reduction which is often needed when communications among sensors are subject to finite capacity constraints. The central question we intend to ask is: if each node in a decentralized inference system has to summarize its data using a finite number of bits, is it still optimal to implement data reduction using global sufficient statistics prior to quantization? We show that the answer is negative using a simple example and proceed to identify conditions under which sufficiency based data reduction followed by quantization is indeed optimal. They include the well known case when the data at decentralized nodes are conditionally independent as well as a class of problems with conditionally dependent observations that admit conditional independence structure through the introduction of an appropriately chosen hidden variable.
\end{abstract}

\section{Introduction}
\IEEEPARstart{A}{}guiding principle for data reduction is the sufficiency principle \cite{fisher:tstatis,Casella:stata,EL:pt_esti}. A sufficient statistic is a function of the data, chosen so that it `should summarize the whole of the relevant information  supplied by the sample' \cite{fisher:tstatis}. A classical example is in binary hypothesis testing where the likelihood ratio can be shown to be a sufficient statistic of the unknown hypothesis, thus can be used instead of the raw data for subsequent decision making \cite{steve:detection}. Another example is the waveform channel with additive white Gaussian channel as often assumed in digital communications \cite{wozen:prin}. It can be easily established that the outputs of simple correlators (or equivalently, that of matched filters) form a sufficient statistic for the unknown input signals. In both examples, the original data, often of high or infinite dimensions, is reduced to low dimension statistics which greatly facilitate the subsequent inference. Indeed, the sufficiency principle has played a prominent role in designing various data processing methods for statistical inference and it encompasses numerous results that have been developed since Fisher's original work. The well-known Neyman-Fisher factorization theorem, for example, provides a systematic way for identifying sufficient statistics using the likelihood function.

This paper studies data reduction in decentralized inference and extends the sufficiency principle to systems where data reduction needs to be done {\em locally}. Decentralized inference refers to the decision making process involving multiple sensors \cite{Radner:team}. Each sensor summarizes its observation and sends a message to the fusion center, which makes the final decision based on the messages it receives. Sensor processing is independent of each other as each sensor has access only to its own data. Illustrated in Fig.~\ref{fig:parallel} is a two-sensor canonical model for decentralized inference where sensors are connected in parallel to a fusion center.
\begin{figure}
    \centerline{
      \psfrag{x1}[c][c]{\scriptsize$\mathbf{X}_1$}
        \psfrag{x2}[c][c]{\scriptsize$\mathbf{X}_2$}
        %\psfrag{T1}[c][c]{\scriptsize$T_1(\Xbf_1)$}
        %\psfrag{T2}[c][c]{\scriptsize$T_2(\Xbf_2)$}
        \psfrag{u1}[c][c]{\tiny$T_1(\Xbf_1)$}
        \psfrag{u2}[c][c]{\tiny$T_2(\Xbf_2)$}
        \psfrag{p}[c][c]{\footnotesize$p(\mathbf{x}|\theta)$}
        %\psfrag{s1}[c][c]{\footnotesize$T_1(\cdot)$}
        %\psfrag{s2}[c][c]{\footnotesize$T_2(\cdot)$}
        \psfrag{g1}[c][c]{\footnotesize$\mathbb{X}_1$}
        \psfrag{g2}[c][c]{\footnotesize$\mathbb{X}_2$}
        \psfrag{h}[c][c]{\footnotesize$h(\cdot)$}
        \psfrag{t}[c][c]{\footnotesize$\theta$}
        \psfrag{t2}[c][c]{\footnotesize$\hat{\theta}$}
        \scalefig{.4}\epsfbox{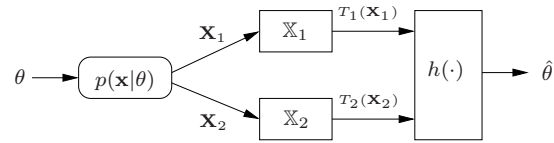}}
        \caption{A parallel network involving two peripheral sensors.}
        \label{fig:parallel}
\end{figure}

For decentralized inference, data reduction is done locally without access to the global data. Therefore, the contrasting notions of local sufficiency and global sufficiency need to be treated with care  \cite{RV:suff}. A sufficient statistic defined with respect to local data is referred to as a local sufficient statistic; if a collection of local statistics form a global sufficient statistic, they are said to be globally sufficient. For the special case when data are conditionally independent given the inference parameter, local sufficient statistics are known to be globally sufficient \cite{RV:suff,EB:fusion,Ishwar:rate}. However, for the general case when data are conditionally dependent, a set of local sufficient statistics need not be globally sufficient and vice versa. The first objective of this paper is to develop theories and tools for decentralized data reduction with conditionally dependent observations for parallel networks. We show that global sufficiency of local statistics is not determined solely by the statistical characterization of local data but also depends on the statistical property of the global data.%\footnote{The structure of the network is also shown to affect the global sufficiency of local statistics in [14]. Its impacts on the relationship between global sufficiency and local sufficiency are not studied as we only consider parallel networks in this paper.}.

Sufficiency based data reduction ensures no loss of inference performance using the reduced data. While the sufficiency principle often results in maximum dimensionality reduction, communicating a one-dimensional real data may still be infeasible when communication is  subject to a finite capacity constraint. In this paper, we consider the simple case where each sensor node communicates only a finite number of bits to the fusion center,  i.e., it needs to summarize its observation using a finite number of bits. Directly quantizing the raw data, especially if the data is of high dimension and quantizers operate in a decentralized fashion, is often a formidable task \cite{Lam:DesignQ,John:Distri}. As such, it is often desirable to achieve maximum data reduction at each node prior to quantization. 

We are then led to the question: \textit{is it optimal to implement data reduction by forming a collection of global sufficient statistics followed by the design of optimal quantizers using the reduced data?} Alternatively, one can consider the sufficiency principle to be the ubiquitous principle for data reduction in a `lossless' sense, that is, complete information in the original data needs to be retained in the statistics. When practical constraints such as finite-bit quantization are imposed which result in inevitable loss of information, is sufficient principle still the guiding principle for data reduction?

Unfortunately, as seen from Example~\ref{exp:counter suff}, the answer to this question is negative in general. However, there exist known results where quantizing sufficient statistics is shown to be optimal. The classical example is distributed detection with conditionally independent observations where the local likelihood ratios form a set of global sufficient statistics. Indeed, Tsitsiklis established in \cite{Tsitsiklis:extremal} that likelihood ratio quantizers (LRQ's) are optimal for a broad class of performance criteria. There also exist instances where quantizing local sufficient statistics is globally optimal for certain parameter regimes in the dependent observation case \cite{Willet:good}. The second objective of this paper is thus to identify, for decentralized inference involving dependent data, conditions under which data reduction using sufficient statistics is still optimal when quantization is required at each node. While the result includes that of \cite{Tsitsiklis:extremal} as its special case, the approach differs from that of \cite{Tsitsiklis:extremal} as we do not start with an explicit form of quantizers thus can not explore the structural information of the statistics as that of \cite{Tsitsiklis:extremal}. Instead, our approach utilizes the Markovian structure implied in sufficient statistics. On the other hand, our optimality is strictly in the sense of minimizing a Bayesian risk as opposed to that of \cite{Tsitsiklis:extremal} which includes a broader class of performance criteria.

Preliminary results were reported in \cite{Ge:sufficiency} and \cite{Shengyu:quant}. In addition to expanding on the technical details of previous studies, the current paper introduces an alternative characterization of structurally optimal data reduction (cf. Section \ref{SC:generalcond}).

Also related to the present work is the quantizer design for distributed estimation in \cite{Lam:DesignQ} and \cite{John:Distri} where necessary conditions for optimal quantizers are derived. The present work does not explicitly address the quantizer design problem. Instead, we derive sufficient conditions such that sufficient statistics based data reduction followed by quantization is structurally optimal. Various optimal quantizer design approaches can then be applied to the reduced data which is often much more tractable than dealing with the raw data.

The rest of the paper is organized as follows. Section \ref{SC:Cen} reviews the basic sufficiency principle for centralized inference, including the optimality of sufficiency based data reduction when quantization is required. Section \ref{SC:sp} deals with data reduction in decentralized inference with conditionally dependent observations in the absence of a quantization constraint. In Section \ref{SC:Decen}, the sufficiency principle is re-examined in decentralized inference when quantization is necessary at each node, i.e., only a finite number of bits can be used to summarize the reduced data at each sensor. Both conditionally independent and conditionally dependent observations are considered. We establish the structural optimality of sufficiency based data reduction followed by quantizers for the independent case. For the dependent case, we identify a class of problems where we prove that sufficiency based data reduction is still optimal in the presence of quantizers. In Section \ref{SC:generalcond}, we obtain a sufficient condition under which the sufficiency based data reduction stills attains the same optimal inference as the raw data. It includes both the independence and dependence conditions as its special cases. Section \ref{conclusion} concludes the paper.

\section{Centralized inference}
\label{SC:Cen}
In this section, we consider a simple centralized inference system where the entire data is available at a single node. We  review the basic sufficiency principle for centralized inference and then establish the optimality of sufficiency based data reduction when quantization is required.
 \subsection{Sufficiency principle}
Suppose $\theta$ is the parameter of inference interest and $\Xbf\triangleq\{X_1,\cdots,X_n\}$ is a random vector observation collected at the node, whose distribution is given by $p(\xbf|\theta)
\footnote{We do not distinguish between probability density and probability mass functions. Its meaning will become clear in the context of specific problems.}$.
 The sufficiency principle states that a function (or statistic) of $\Xbf$, denoted by $T(\Xbf)$, is  a sufficient statistic for $\theta$ if the inference outcome does not change when either $\xbf$ or $\ybf$ is observed as long as $T(\xbf)=T(\ybf)$ \cite{Casella:stata}.
%If $T(\Xbf)$ is a sufficient statistic for $\theta$, then any inference about $\theta$ should depend on $\Xbf$ only through  $T(\Xbf)$\cite{Casella&Berger:book}.
A useful tool to identify sufficient statistics is the Neyman-Fisher factorization theorem\cite{Casella:stata} which states that  a statistic  $T(\Xbf)$ is sufficient for $\theta$ if and only if there exist functions $g(t|\theta)$ and $h(\xbf)$ such that %
 \bqa p(\xbf|\theta)=g(T(\xbf)|\theta)h(\xbf).\label{NF factor}\eqa
If the parameter $\theta$ is itself random, the sufficiency principle can  be elegantly reframed using the data processing inequality, assisted with the use of Shannon's mutual information
\cite{Cover:Information}. That is, a function $T(\Xbf)$ is a sufficient statistic if and only if the following Markov chain holds
\bqa \label{eqn:mc for ss} \theta-T(\Xbf)-\Xbf,\eqa
which is equivalent to the mutual information equation
\bqa I(\theta;\Xbf)=I(\theta,T(\Xbf)).\eqa
The following lemma, which is used throughout the paper, is a straightforward result from the definition of Markov chain.
\begin{lemma}
\label{prob}
Let $\Xbf\sim p(\xbf|\theta)$ where $\theta$ is a random parameter. If $T(\Xbf)$ is a sufficient statistic for $\theta$ with respect to $\Xbf$, then
$$p(\theta|\mathbf{x})=p(\theta|T(\mathbf{x})).$$
\end{lemma}
\begin{proof}
As $T(\Xbf)$ is a function of $\Xbf$, $\theta-\Xbf-T(\Xbf)$ form a Markov chain. Together with (2), we have $$p(\theta|\xbf)=p(\theta|\xbf,T(\xbf))=p(\theta|T(\xbf)).$$% form a Markov chain and then $p(\theta|\mathbf{x},T(\xbf))=p(\theta|T(\mathbf{x}))$ according to the definition of Markov chain. The proof is thus complete by noting that $T(\mathbf{X})$ is functionally dependent on $\mathbf{X}$.
\end{proof}

\subsection{Centralized inference with quantization}

Consider a centralized inference system in which quantization is required, as shown in Fig.~\ref{fig:centralss1}. Here, $\theta$ is the parameter of inference interest with distribution $p(\theta)$, $\mathbf{X}$ is the random vector observation, $\gamma(\cdot)$ is the quantizer directly operating on the data $\Xbf$ and the output of the quantizer is $U=\gamma(\mathbf{X}) \in \{0,\dots,L-1\}$ where $L$ is the number of possible outputs. The estimator at the fusion center is denoted by the function $h(\cdot)$ whose input is the quantizer output.

  Let $T(\mathbf{X})$ be any sufficient statistic for $\theta$. To establish the optimality of sufficiency based data reduction with a quantization constraint, we need to investigate whether the two systems in Fig.~\ref{fig:centralss} achieve the same optimal performance where the second system applies data reduction to obtain $T(\Xbf)$ prior to a quantization operation. The quantizer and estimator in Fig.~\ref{fig:centralss2} are similarly defined by $U'=\gamma'(T(\mathbf{X}))$ and $h'(U')$.  Note that for a centralized system there is no distinction between local and global sufficient statistics.
\begin{figure}
    \centering
    \subfigure[]
    {
        \psfrag{x}[c][c]{\scriptsize$\mathbf{X}$}
        %\psfrag{Tx}[c][c]{\scriptsize$T(\mathbf{X})$}
        \psfrag{u}[c][c]{\scriptsize$U$}
        \psfrag{p}[c][c]{\footnotesize$p(\mathbf{x}|\theta)$}
        %\psfrag{ss}[c][c]{\footnotesize$T(\cdot)$}
        \psfrag{g}[c][c]{\footnotesize$\gamma(\cdot)$}
        \psfrag{h}[c][c]{\footnotesize$h(\cdot)$}
        \psfrag{t}[rc][c]{\footnotesize$\theta$}
        \psfrag{t2}[c][c]{\footnotesize$\hat{\theta}$}
        \scalefig{.34}\epsfbox{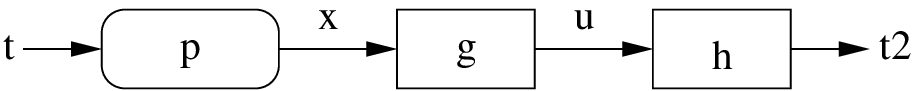}
        \label{fig:centralss1}
    }
    \\
    \subfigure[]
    {
        \psfrag{x}[c][c]{\scriptsize$\mathbf{X}$}
        \psfrag{tx}[c][c]{\scriptsize$T(\mathbf{X})$}
        \psfrag{u}[c][c]{\scriptsize$U'$}
        \psfrag{p}[c][c]{\footnotesize$p(\mathbf{x}|\theta)$}
        \psfrag{ss}[c][c]{\footnotesize$T(\cdot)$}
        \psfrag{g}[c][c]{\footnotesize$\gamma'(\cdot)$}
        \psfrag{h}[c][c]{\footnotesize$h'(\cdot)$}
        \psfrag{t}[rc][c]{\footnotesize$\theta$}
        \psfrag{t2}[c][c]{\footnotesize$\hat{\theta}'$}
        \scalefig{.45}\epsfbox{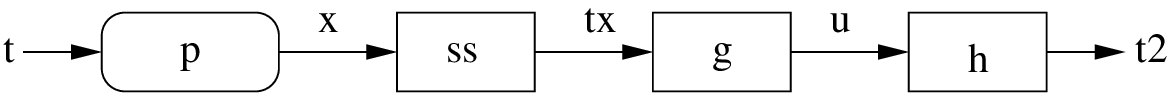}
        \label{fig:centralss2}
    }
    \caption{Centralized inference systems with quantizers operating on (a) the raw data $\Xbf$, (b) the sufficient statistic $T(\Xbf)$.}
    \label{fig:centralss}
\end{figure}

Let $d[\theta,\hat{\theta}]$ be a given cost function between the parameter $\theta$ and the estimator output $\hat{\theta}$. For the model in Fig.~\ref{fig:centralss1}, $\hat{\theta}=h(U)=h(\gamma(\mathbf{X}))$. The Bayesian risk is the expected cost function given by
\bqa
\label{eqn:risk11} 
R=E\{d[\theta,h( \gamma(\Xbf))]\},
\eqa
where the expectation is taken with respect to both the random parameter $\theta$ and the observation $\Xbf$. For the model in Fig.~\ref{fig:centralss2}, $\hat{\theta}'=h'(\gamma'(T(\Xbf))$ and
the Bayesian risk is given by
\bqa
\label{eqn:risk12}
R'=E\{d[\theta,h'(\gamma'(T(\Xbf))]\},
\eqa
where again the expectation is taken with respect to $\theta$ and $\Xbf$. We now establish that the system described in Fig.~\ref{fig:centralss2} is structurally optimal, i.e., it can achieve the same inference performance as that of Fig.~\ref{fig:centralss1}, hence quantizing the sufficient statistic achieves the same minimum Bayesian risk as quantizing the observation in centralized inference.
\begin{theorem}
\label{thmcen}
For the Bayesian risks in (\ref{eqn:risk11}) and (\ref{eqn:risk12}), $$\min_{\gamma,h}R=\min_{\gamma',h'}R'.$$
\end{theorem}
\begin{proof}
Let $$R_\mathrm{min}=\min_{\gamma,h}R.$$ Denote by $\gamma^*(\cdot)$ and $h^*(\cdot)$ the optimal quantizer and estimator that achieve $R_\mathrm{min}$ is achieved. Apparently, $R'\geq R_\mathrm{min}$ as one can always define a new quantizer
$\gamma(\Xbf)=\gamma'(T(\Xbf))$ for any given $\gamma'(\cdot)$, thus converting any system described by Fig.~\ref{fig:centralss2} to that of Fig.~\ref{fig:centralss1} whose performance is no better than $R_{\mathrm{min}}$. Then we only need to show that there exist $\gamma'(\cdot)$ and $h'(\cdot)$ such that the corresponding $R'=R_\mathrm{min}$.

Expanding $R$ in (\ref{eqn:risk11}) with respect to the observation $\mathbf{X}$, we have
%\begin{itemize}
%	\item Continuous parameter of inference interest: $\theta$ with probability $p(\theta)$.
%	\item Sensor observation $\mathbf{X}$.
%	\item Quantizer output $U=\gamma(\mathbf{X}) \in \{0,1\}$.
%	\item Fusion center output $\hat{\theta}=h(U)=h(\gamma(\mathbf{X})).$
%\end{itemize} If $T(\mathbf{X})$ is a sufficient statistic for $\theta$, then any inference about $\theta$ depending on $\mathbf{X}$ is only through $T(\mathbf{X})$.

%Let $d$ be the distortion function between $\theta$ and the output of the centralized estimation system which is $h(\gamma(\mathbf{X}))$. The error expression that needs to be minimized is given by
\begin{align}
\label{eqn:errorcentral}
R&= \int_{\theta}\int_{\mathbf{x}}d[\theta,h(\gamma(\xbf))]p(\mathbf{x},\theta)d\mathbf{x}d\theta\nn\\
&= \int_{\mathbf{x}}\int_{\theta}d[\theta,h(\gamma(\xbf))]p(\theta|\mathbf{x})p(\mathbf{x})d\theta d\mathbf{x}\nn\\
&\stackrel{(a)}{=} \int_{\mathbf{x}}\int_{\theta}d[\theta,h(\gamma(\xbf))]p(\theta|T(\mathbf{x}))p(\mathbf{x})d\theta d\mathbf{x}\nn\\
&\triangleq \int_{\mathbf{x}}\alpha(u,\mathbf{x})p(\mathbf{x})d\mathbf{x},
\end{align}
where
\begin{align}
u&\triangleq\gamma(\xbf),\nn\\
\alpha(u,\mathbf{x})&\triangleq\int_{\theta}d[\theta,h(u)]p(\theta|T(\mathbf{x}))d\theta,\nn
\end{align}
and ($a$) is from Lemma \ref{prob}. From (\ref{eqn:errorcentral}), $\gamma^*(\cdot)$ must be such that it chooses $u$ to minimize $\alpha(u,\xbf)$ with $h^*(\cdot)$ used as the estimator, that is
\bqa
\label{eqn:quandec}
U=\gamma^*(\mathbf{X})=\arg\min\limits_{u}\alpha(u,\mathbf{X}).\
\eqa
Therefore, $u=\gamma^*(\mathbf{x})=s \in \{0,\dots,L-1\}$ given the observation $\mathbf{x}$ if for any $t \in \{0,\dots,L-1\}$,
\begin{align}
\label{eqn:centralcondition}
0 &\geq\alpha(s,\mathbf{x})-\alpha(t,\mathbf{x})\nn\\
&=\int_{\theta}\{d[\theta,h^*(s)]-d[\theta,h^*(t)]\}p(\theta|T(\mathbf{x}))d\theta.
%&=&\int_{\theta}\{d[\theta,\theta_0]-d[\theta,\theta_1]\}p(\theta|T(\mathbf{x}))d\theta.
\end{align}
Given $h^*(\cdot)$, (\ref{eqn:centralcondition}) is a necessary condition for $\gamma^*(\cdot)$ to achieve $R_\mathrm{min}$. Note that (\ref{eqn:centralcondition}) depends on $\Xbf$ only through $T(\Xbf)$, hence can be realized by a $\gamma'(T({\Xbf}))$. If we use this $\gamma'(\cdot)$ together with $h'(\cdot)=h^*(\cdot)$, then $R'=R_\mathrm{min}$.
\end{proof}
%{\em Remarks:}

The above result is not surprising in view of the fact that a sufficient statistic captures all the information about $\theta$ contained in the data. Indeed, the above theorem can be viewed as a simple instantiation of the sufficiency principle for the Bayesian risk. Other inference objective functions can also be used. Consider for example the ``indirect rate distortion problem" \cite{HS:indirect} where a noisy version of a source sequence is observed at the encoder while the decoder tries to  minimize the end-to-end distortion subject to a rate constraint between the encoder and the decoder. It was shown in \cite{KW:rateloss} that data reduction using a sufficient statistic at the encoder does not affect the rate distortion function.

In decentralized inference, however, the same statement is not necessarily true, i.e., sufficient statistics based data reduction may not be optimal when quantization is required at individual nodes. Before we study the impact of quantizers on data reduction in a decentralized system, we first revisit the sufficiency principle when no quantization is required. In particular, we strive for a better understanding of the relationship between local and global sufficient statistics under various dependent models.

%Clearly, from (\ref{eqn:quandec}), the same result holds for any finite number of bits as the quantizer chooses an output that minimizes that corresponding $\alpha(\cdot)$ function. For simplicity and ease of notation, we again impose the constraint that quantizer output is binary in the following sections.

\section{Sufficient Statistics in Decentralized Inference}
\label{SC:sp}
This section considers the decentralized data reduction in a two-sensor parallel network  as illustrated in Fig.~\ref{fig:parallel}. The results extend naturally to the case with an arbitrary number of sensors. Let $\theta\sim p(\theta)$ be the parameter of interest and $\Xbf_i$ the local observation at sensor $i$ for $i=1,2$.
 For a decentralized system, there is a need to distinguish between the notions of local and global sufficient statistics \cite{RV:suff}. %For the system described in Fig.~\ref{fig:decentralss2}
When $\theta$ is random,  for $i=1,2$, $T_i(\Xbf_i)$ is a local sufficient statistic if
\bqa
\label{eqn:markov1}
\theta-T_i(\Xbf_i)-\Xbf_i
\eqa
form a Markov chain,  i.e., sufficiency is defined with respect to the local observation $\Xbf_i$. On the other hand, we call $(T_1(\Xbf_1),T_2(\Xbf_2))$ a global sufficient statistic if
\bqa
\label{eqn:markov2}
\theta - (T_1(\Xbf_1), T_2(\Xbf_2))- (\Xbf_1,\Xbf_2)
\eqa
form a Markov chain. It is apparent that for the general case, the two individual Markov chains (\ref{eqn:markov1}) and (\ref{eqn:markov2}) do not imply each other.

\subsection{Conditionally Independent Observations}
For the conditional independence case, it can be easily established that local sufficiency implies global sufficiency \cite{RV:suff,EB:fusion,Ishwar:rate}. The converse also holds for the conditional independence case, which is given in the following proposition.
\begin{proposition}
Let $\Xbf_1$ and $\Xbf_2$ be conditionally independent observations given the random parameter $\theta$. If $(T_1(\Xbf_1), T_2(\Xbf_2))$ form a global sufficient statistic for $\theta$, then both $T_1(\Xbf_1)$ and $T_2(\Xbf_2)$ are respectively local sufficient statistics with respect to the observations $\Xbf_1$ and $\Xbf_2$. 
\end{proposition}

We first state some useful properties of Markov chains  \cite{Pearl:causality} that will be used for subsequent proofs:
\begin{itemize}
\item Symmetry: $X-Z-Y\Rightarrow Y-Z-X$;
\item Decomposition: $X-Z-YW\Rightarrow X-Z-Y$;
\item Weak Union:$X-Z-YW\Rightarrow X-ZW-Y$;
\item Contraction: $X-Z-Y$ and $X-ZY-W \Rightarrow X-Z-YW$;
\item Intersection: $X-ZW-Y$ and $X-ZY-W \Rightarrow X-Z-YW$.
\end{itemize}

\begin{proof} 
Since $\Xbf_1$ and $\Xbf_2$ are independent given $\theta$, $\Xbf_1-\theta-\Xbf_2$ form a Markov chain and so does $(\Xbf_1, T_1(\Xbf_1))-\theta-\Xbf_2$ as $T_1(\Xbf_1)$ is a function of $\Xbf_1$. Using the weak union property, we have that $\Xbf_1-(\theta, T_1(\Xbf_1))-\Xbf_2$ form a Marokov chain. That $(T_1(\Xbf_1), T_2(\Xbf_2))$ is globally sufficient implies that (\ref{eqn:markov2}) holds and thus $\Xbf_1-(T_1(\Xbf_1), T_2(\Xbf_2))-\theta$ form a Markov chain according to the decomposition and symmetry properties. Combining $\Xbf_1-(\theta, T_1(\Xbf_1))-\Xbf_2$ and $\Xbf_1-(T_1(\Xbf_1), T_2(\Xbf_2))-\theta$, and using the intersection property we get the Markov chain $\Xbf_1-T_1(\Xbf_1)-(\theta, T_2(\Xbf_2))$ whenever $p(\xbf_1,T_1(\xbf_1),T_2(\xbf_2),\theta)$ is positive. Thus $T_1(\Xbf_1)$ is a local sufficient statistic for $\theta$. That $T_2(\Xbf_2)$ is locally sufficient for $\theta$ can be established similarly.
\end{proof}

%Indeed, $\Xbf_1-(\theta, T_1(\Xbf_1))-T_2(\Xbf_2)$ and $\Xbf_1-(T_1(\Xbf_1),T_2(\Xbf_2))-\theta$ are two Markov chains, which yield that $\Xbf_1-T_1(\Xbf_1)-(\theta,T_2(\Xbf_2))$ also form a Markov chain whenever $p(\xbf_1,T_1(\xbf_1),T_2(\xbf_2),\theta)$ is positive. Thus, $T_1(\Xbf_1)$ is a local sufficient statistic for $\theta$. That $T_2(\Xbf_2)$ is locally sufficient for $\theta$ can be shown similarly.
\subsection{Conditionally Dependent Observations}
While the above establishes that global and local sufficient statistics imply each other for conditionally independent observations, the same is not true for the dependent case. Consider the following trivial example.
\begin{example}
Let $\Xbf_1=\Xbf_2$ in Fig.~\ref{fig:parallel}. It is clear that $(T_1(\Xbf_1)=\Xbf_1,T_2(\Xbf_2)=\varnothing)$ is globally sufficient for $\theta$ while $T_2(\Xbf_2)=\varnothing$ is not locally sufficient.
\end{example}

 % decentralized inference becomes considerably complex. Global sufficiency of local statistics is not determined solely by the statistical characterization of local data but also depends on the statistical property of the global data as well as the structure of the network.
% As $T_i(\Xbf_i)$ is only a function of local data $\Xbf_i$, there is no guarantee that a pair of local sufficient statistics will form a global sufficient statistic. %Various results have been obtained to allow one to construct global sufficient statistics at decentralized nodes \cite{Hall&Wessel&Wise:91IT,RV:suff,Ishwar:rate}. The most notable result is that with conditionally independent observations, local sufficient statistics are also globally sufficient.

The rest of this section is devoted to the question of how to identify global sufficient statistics at distributed nodes with conditionally dependent observations. Our approach leverages a recently proposed hierarchical conditional independence (HCI) model, which is a new framework developed for distributed detection with conditionally dependent observations \cite{Hchen:newframework}. %The idea of the model is to introduce a hidden variable such that the sensor observations are conditionally independent with respect to this new variable regardless of the dependent structure of the original model.%
An HCI model is constructed by introducing a hidden variable $\Wbf$ such that the following Markov chains hold: % $\theta-\Wbf-(\Xbf_1,\Xbf_2)$ and $\Xbf_1-\Wbf-\Xbf_2$.
 \bqa \begin{array}{ll} \label{HCI}
&\Xbf_1-\Wbf-\Xbf_2,\\
&\theta-\Wbf-(\Xbf_1,\Xbf_2).
\end{array}\eqa
That is, $\Wbf$ induces conditional independence between $\Xbf_1$ and $\Xbf_2$ as well as conditional independence between the inference parameter $\theta$ and the sensor observations $(\Xbf_1,\Xbf_2)$. It was established in \cite{Hchen:newframework} that any general distributed inference model is equivalent to an HCI model and vice versa. We notice here that while we only illustrate the HCI model using the two sensor system, the framework is applicable to that involving any arbitrary number of sensors where we replace the Markov chain $\Xbf_1-\Wbf-\Xbf_2$ with the equivalent conditional independence assumption.

Notice that the second Markov chain in defining the HCI model implies that the information about the inference parameter $\theta$ in the data $(\Xbf_1,\Xbf_2)$ is preserved entirely in $\Wbf$. This is formalized in the following lemma.
%Assume the parameter $\theta$ is random and let data available at node $\mathbb{X}_i$ be $\Xbf_i$ for $i=1,2$. $(T_1(\Xbf_1),T_2(\Xbf_2))$ are globally sufficient for $\theta$ if the Markov chain $\theta-(T_1(\Xbf_1),T_2(\Xbf_2))-(\Xbf_1,\Xbf_2)$ holds.

%Identifying local statistics that are globally sufficient can be accomplished in theory via the factorization theorem. The process of using the factorization theorem may become cumbersome in a decentralized system or not applicable when the precise joint distribution of the data in the network is not available at local nodes. The following theorem provides certain relation between local sufficiency and global sufficiency
%for a class of distributed inference problem.
%gives the relation of local sufficiency and global sufficiency under
%facilitate the process to find globally sufficient statistics.

\begin{lemma}\label{lemma 1}
Let $\Xbf_1,\Xbf_2\sim p(\xbf_1,\xbf_2|\theta)$ and suppose that there exists a random variable $\Wbf$ such that
\bqa \theta-\Wbf-(\Xbf_1,\Xbf_2).\label{eq 1}\eqa
A statistic $T(\Xbf_1,\Xbf_2)$ that is sufficient  for $\Wbf$ is also sufficient for $\theta$.
\end{lemma}
\begin{proof}
The Markov chain (\ref{eq 1}) implies that  $\theta-\Wbf-(\Xbf_1,\Xbf_2,T(\Xbf_1,\Xbf_2))$  forms a Markov chain for any statistics $T(\Xbf_1,\Xbf_2)$.
That $T(\Xbf_1,\Xbf_2)$ is sufficient  for $\Wbf$ implies the Markov chain $\Wbf-T(\Xbf_1,\Xbf_2)-(\Xbf_1,\Xbf_2)$.
It is straightforward to show that these two Markov chains give rise to a  long Markov chain
\bqn\theta-\Wbf-T(\Xbf_1,\Xbf_2)-(\Xbf_1,\Xbf_2).\eqn
Therefore,  $T(\Xbf_1,\Xbf_2)$ is sufficient for $\theta$.
\end{proof}

Lemma \ref{lemma 1} is not useful in itself as $T(\Xbf_1,\Xbf_2)$ is a function of the global data which is not available in either of the nodes. Its use is mainly for establishing the following result.

\begin{theorem}\label{theorem 1}
Let $\Xbf_1,\Xbf_2\sim p(\xbf_1,\xbf_2|\theta)$ and suppose there exists a random variable $\Wbf$ such that
$\theta-\Wbf-(\Xbf_1,\Xbf_2)$.
Let $T(\Wbf)$ be a sufficient statistic for $\theta$, i.e., $\theta-T(\Wbf)-\Wbf$.
\begin{enumerate}
\item If a pair of statistics $(T_1(\Xbf_1),T_2(\Xbf_2))$ are globally sufficient  for $T(\Wbf)$,
they are globally sufficient  for $\theta$.

\item If $T(\Wbf)$ induces conditional independence between $\Xbf_1$ and $\Xbf_2$ and $(T_1(\Xbf),T_2(\Xbf_2))$ are locally sufficient for $T(\Wbf)$, then $(T_1(\Xbf_1),T_2(\Xbf_2))$ are globally sufficient for $\theta$.
\end{enumerate}
\end{theorem}
\begin{proof} To prove 1), from Lemma \ref{lemma 1}, we  only need to show that $\theta-T(\Wbf)-(\Xbf_1,\Xbf_2)$ holds. Note first that $T(\Wbf)-(\theta,\Wbf)-(\Xbf_1,\Xbf_2)$ form a Markov chain as $T(\Wbf)$ is a function of $\Wbf$ . Together with $\theta-\Wbf-(\Xbf_1,\Xbf_2)$ we obtain the Markov chain $(\theta,T(\Wbf))-\Wbf-(\Xbf_1,\Xbf_2)$ using the contraction property.
Combined with the Markov chain $\theta-T(\Wbf)-\Wbf$, we get $\theta-T(\Wbf)-\Wbf-(\Xbf_1,\Xbf_2)$ which implies $\theta-T(\Wbf)-(\Xbf_1,\Xbf_2)$.

To prove 2), since conditional independence ensures that local sufficient statistics are globally sufficient, $(T_1(\Xbf_1),T_2(\Xbf_2))$ are thus sufficient for $T(\Wbf)$. The result in 1) thus establishes that they are also sufficient for $\theta$.
\end{proof}

Applying Theorem \ref{theorem 1} to the HCI model, we have the following corollary.

\begin{corollary}
\label{cor1}
For an HCI model, local sufficiency with respect to the hidden variable implies global sufficiency.
\end{corollary}

Corollary 1 suggests that a way to obtain global sufficient statistics at individual nodes is to ensure local sufficiency of the statistics with respect to the hidden variable $\Wbf$ in the HCI model. As we shall illustrate, the approach is meaningful only if the hidden variable $\Wbf$ is chosen appropriately. For example, choosing $\Wbf=(\Xbf_1,\Xbf_2)$ ensures that the Markov chains used in defining the HCI model in (8) are always satisfied, yet it does not lead to any meaningful data reduction. 
\subsection{Examples}
\label{SC:exp}
We now use a simple example to show how Corollary 1 can be used for data reduction through an appropriately chose $\Wbf$.%There does not exist a general procedure for choosing a meaningful $\Wbf$. Nevertheless, one may find $\Wbf$ either from network structures or characterizations of local statistics in many cases. Several examples have been given in \cite{Hchen:newframework}. 

%A meaningful hidden variable may be found from the structure of the network or the characterizations of local statistics, but there is not a general procedure for finding one. The following is a simple example where we can choose the meaningful hidden variable $\Wbf$ via the characterizations of observations and we further use Corollary 1 to achieve data reduction through the appropriately chosen $\Wbf$.

%{\em Remarks:}
%\bi
%\item
%Given a local sufficient statistic $T_2(\Xbf_2)$, it is possible that there does not exist a $T_1(\Xbf_1)$ forming a globally sufficient statistic
%together with $T_2(\Xbf_2)$.
%\item
% The above result is shown under the assumption that $\theta$ is a random variable, similar result can be obtained for $\theta$ is not random by %resorting to factorization theorem instead of data processing inequality.
% \ei

\begin{example}
\label{exp:globalsuff}
For $i=1,\cdots,n$, let
\bqa
X_{1i}&=&\theta+Z+U_i,\nn\\
X_{2i}&=&\theta+Z+V_i,\nn
\eqa
where $\theta, Z, U_1, \cdots, U_n, V_1, \cdots, V_n$ are mutually independent Gaussian random variables such that $\theta\sim \mathcal{N}(0,1)$, $Z\sim \mathcal{N}(0,\rho)$, $U_j\sim \mathcal{N}(0,1-\rho)$, $V_j\sim \mathcal{N}(0,1-\rho)$. Thus, we need to estimate a parameter $\theta$ in the presence of a constant interference $Z$ and independent noises $U_i$ and $V_i$. Since $X_{1i},X_{2i}\sim N(\theta,\theta,1,1,\rho)$, given $\theta$ $\Xbf_1$ and $\Xbf_2$ are not independent conditioned on $\theta$.

Choose the hidden variable  $W=\theta+Z$. One can verify easily that $W$ satisfies the Markov chains $\theta-W-(\Xbf_1,\Xbf_2)$ and $\Xbf_1-W-\Xbf_2$ required by the HCI model. For Gaussian observations, it is also clear that $\sum_i{X_{1i}}$ and $\sum_i{X_{2i}}$ are locally sufficient for $W$. Therefore, from Corollary \ref{cor1}, $(\sum_iX_{1i},\sum_iX_{2i})$ is globally sufficient for $\theta$. %and quantizing $\sum_i{X_{1i}}$ and $\sum_i{X_{2i}}$ is structurally optimal by Theorem \ref{Theorem Quant CD}.
\end{example}
%\vspace{0.05in}

While the above example is somewhat artificial, it does provide a clue for how to choose a meaningful hidden variable - often times, the signal model itself provides a natural choice of $\Wbf$ as in Example 2. Specifically, the signal plus interference term $\theta+Z$ satisfies both Markov chain conditions and turns out to be precisely the hidden variable that leads to meaningful data reduction. The next example is motivated by the cooperative spectrum sensing problem \cite{Peng:energydetection}. As with Example 2, the choice of $\Wbf$ can also be obtained by careful examination of the signal model.
\begin{example}\label{eg:cooperative}
Consider the hypothesis testing problem involving $K$ sensors with the two hypotheses under test given by
\begin{align}
H_i&: X_k=h_kS+N_k, i = 0,1,\nn
\end{align}
where $X_k$, $k=1,\cdots,K$,  is the observation at sensor $k$, $h_k$'s are circularly symmetric complex Gaussian and independent of each other and of other variables, $S$ is a signal taking values in $\mathcal{S}_0=\{s_0=0\}$ under $H_0$ and
$\mathcal{S}_1=\{s_m=r_me^{j\theta_m},m=1,\cdots,M\}$ with probability $p(S=s_m)=\pi_m$ under $H_1$,  and $N_k$ is the observation noise at the $k$th sensor which is circularly complex Gaussian distributed and is independent of each other. This hypothesis testing problem can be used to describe the baseband model of detecting the presence of a QAM signal in independent Rayleigh fading channels using $K$ sensors.
%Each sensor makes a local decision $U_k=\gamma(X_k)$ and sends it to a fusion center which makes a final decision regarding the hypothesis under test.

The observations are not conditionally independent under $H_1$ given that the observations contain a common random signal $S$.
Again, taking a Bayesian viewpoint where we assume that the true hypothesis $H$ is a binary random variable, then $H-S-(X_1,\cdots,X_K)$ form a Markov chain since the observations depend on the hypothesis only through the signal.
 It is easy to verify that the statistic $|S|$ is sufficient for $H$ given $S$.
Thus, the Markov chain $H-|S|-S-(X_1,\cdots,X_K)$ holds.
On the other hand, given $|S|$,  the observations are conditionally independent of each other under the independent Rayleigh fading assumption.
Therefore, $|S|$ serves as the hidden variable $\Wbf$ for the HCI model corresponding to this decentralized hypothesis testing problem.

For any $k$,
 $|X_k|$ is a minimal sufficient statistic for $|S|$. This can be easily verified by writing out the ratio $\frac{p(x_k||s|)}{p(x_k'||s|)}$ for
two sample points $x_k $ and $x_k'$.
Therefore, from Corollary \ref{cor1}, $\{|X_k|\}, k=1,\cdots, K$, are globally sufficient for $H$.

%Also, by the above discussion, we can establish that quantizing the statistic $|X_k|$ at the $k$th is structurally optimal by Theorem \ref{Theorem Quant CD}. Therefore, the optimal detector at each local sensor is an energy detector \cite{Peng&Chen&Chen:12CISS}.

\end{example}

\section{Decentralized Data Reduction With Quantization Constraints}
\label{SC:Decen}
We now consider decentralized inference where quantization is required at each node. For simplicity and ease of presentation, we again assume a simple two-node system, as illustrated in Fig.~\ref{fig:decentralss}. The result extends to systems with more than two nodes in a straightforward manner.

Let $\theta\sim p(\theta)$ be the parameter of interest and $\Xbf_i$ the local observation at sensor $i$ with a likelihood function $p(\xbf_i|\theta)$, for $i=1,2$. Statistics and quantizers at local nodes, as well as the estimator at the fusion center are defined in a similar fashion as that in Section \ref{SC:Cen}. Let $d[\theta,\hat{\theta}]$ be the cost function where $\theta$ is the true parameter and $\hat{\theta}$ its estimate. The Bayesian risks for the systems in Fig. \ref{fig:decentralss1} and Fig. \ref{fig:decentralss2} are given respectively by

\bqa
\label{eqn:risks21}
&R=E\{d[\theta,h(U_1,U_2)]\}
\eqa
and
\bqa
\label{eqn:risks22}
&R'=E\{d[\theta,h'(U_1',U_2')]\},
\eqa
where $U_i=\gamma_i(\Xbf_i) \in \{0,\dots,L-1\}$ and $U_i'=\gamma_i'(T_i(\Xbf_i))\in \{0,\dots,L-1\}$ .

\begin{figure}
    \centering
    \subfigure[]
    {
      \psfrag{x1}[c][c]{\scriptsize$\mathbf{X}_1$}
        \psfrag{x2}[c][c]{\scriptsize$\mathbf{X}_2$}
        \psfrag{T1}[c][c]{\scriptsize$T_1(\Xbf_1)$}
        \psfrag{T2}[c][c]{\scriptsize$T_2(\Xbf_2)$}
        \psfrag{u1}[c][c]{\scriptsize$U_1$}
        \psfrag{u2}[c][c]{\scriptsize$U_2$}
        \psfrag{p}[c][c]{\footnotesize$p(\mathbf{x}|\theta)$}
        \psfrag{s1}[c][c]{\footnotesize$T_1(\cdot)$}
        \psfrag{s2}[c][c]{\footnotesize$T_2(\cdot)$}
        \psfrag{g1}[c][c]{\footnotesize$\gamma_1(\cdot)$}
        \psfrag{g2}[c][c]{\footnotesize$\gamma_2(\cdot)$}
        \psfrag{h}[c][c]{\footnotesize$h(\cdot)$}
        \psfrag{t}[c][c]{\footnotesize$\theta$}
        \psfrag{t2}[c][c]{\footnotesize$\hat{\theta}$}
        \scalefig{.35}\epsfbox{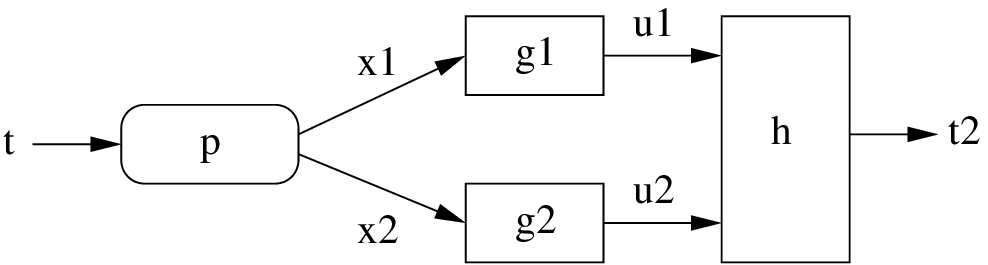}
        \label{fig:decentralss1}
    }
    \\
    \subfigure[]
    {
      \psfrag{x1}[c][c]{\scriptsize$\Xbf_1$}
        \psfrag{x2}[c][c]{\scriptsize$\Xbf_2$}
        \psfrag{T1}[c][c]{\scriptsize$T_1(\Xbf_1)$}
        \psfrag{T2}[c][c]{\scriptsize$T_2(\Xbf_2)$}
        \psfrag{u1}[c][c]{\scriptsize$U_1'$}
        \psfrag{u2}[c][c]{\scriptsize$U_2'$}
        \psfrag{p}[c][c]{\footnotesize$p(\mathbf{x}|\theta)$}
        \psfrag{s1}[c][c]{\footnotesize$T_1(\cdot)$}
        \psfrag{s2}[c][c]{\footnotesize$T_2(\cdot)$}
        \psfrag{g1}[c][c]{\footnotesize$\gamma_1'(\cdot)$}
        \psfrag{g2}[c][c]{\footnotesize$\gamma_2'(\cdot)$}
        \psfrag{h}[c][c]{\footnotesize$h'(\cdot)$}
        \psfrag{t}[c][c]{\footnotesize$\theta$}
        \psfrag{t2}[c][c]{\footnotesize$\hat{\theta}'$}
        \scalefig{.45}\epsfbox{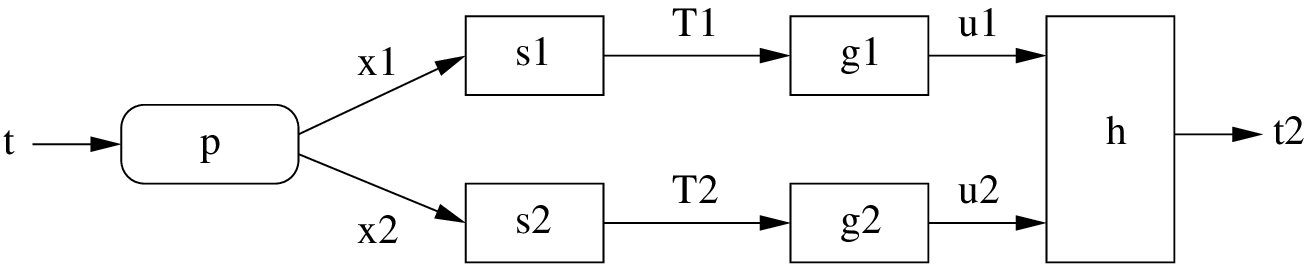}
        \label{fig:decentralss2}
    }
    \caption{Decentralized inference systems with quantizers operating on (a) the raw data $\mathbf{X}_i, i=1,2$, (b) the sufficient statistics $T_i(\Xbf_i), i=1,2$.}
    \label{fig:decentralss}
\end{figure}

%For a decentralized system, there is a need to distinguish the notions of local versus global sufficient statistics \cite{RV:suff}. For the system described in Fig.~\ref{fig:decentralss2} where $\theta$ is random, $T_i(\Xbf_i)$ is a local sufficient statistic if $\theta-T_i(\Xbf_i)-\Xbf_i$, i.e., sufficiency is defined with respect to local observation $\Xbf_i$. On the other hand, we call $(T_1(\Xbf_1),T_2(\Xbf_2))$ a global sufficient statistic if $\theta - (T_1(\Xbf_1), T_2(\Xbf_2))- (\Xbf_1,\Xbf_2$). As $T_i(\Xbf_i)$ is only a function of local data $\Xbf_i$, there is no guarantee that a pair of local sufficient statistics will form a global sufficient statistic.

%Various results have been obtained to allow one to construct global sufficient statistics at decentralized nodes \cite{RV:suff,EB:fusion,Ishwar:rate}. The most notable result is that with conditionally independent observations, local sufficient statistics are also globally sufficient. For the dependent case, there has also been recent effort in identifying global sufficient statistics for decentralized inference \cite{Ge:sufficiency}.

The additional constraint that a quantizer is used at each sensor node may lead to inevitable information loss. As such, it is not clear whether global sufficient statistics based data reduction is still optimal. That is, even if $(T_1(\Xbf_1),T_2(\Xbf_2))$ form a global sufficient statistic, can the system in Fig.~\ref{fig:decentralss2} achieve the same performance as that of Fig.~\ref{fig:decentralss1}?

The answer, unfortunately, is {\em no}, as can be seen from the following simple example.
\begin{example}
\label{exp:counter suff}
Consider the degenerate case where $\Xbf_1=\Xbf_2$ and $U_i$ is constrained to be of one bit. Clearly $(T_1(\Xbf_1)=\Xbf_1,T_2(\Xbf_2)=\varnothing)$ is a global sufficient statistic. However it is trivial to see that quantizing such constructed $T_1(\Xbf_1)$ and $T_2(\Xbf_2)$ using $1$-bit each can be suboptimal compared with quantizing the data directly, with the former equivalent to a $1$-bit quantizer of the data whereas the latter a $2$-bit quantizer. Specifically, for the latter case, $\Xbf_1$ and $\Xbf_2$ are quantized at each node thus provide more information to the fusion center. 
\end{example}

An acute reader has probably realized that the above example involves data that are conditionally dependent given the parameter of interest. It turns out when data are conditionally independent given $\theta$, the answer is indeed the affirmative, i.e., quantizing sufficient statistics is structurally optimal.
\subsection{Conditionally Independent Observations}
\begin{theorem}
\label{thm:cond inde}
For the Bayesian risks in (\ref{eqn:risks21}) and (\ref{eqn:risks22}) when $\Xbf_1$ and $\Xbf_2$ are conditionally independent given $\theta$,
$$\min_{\gamma_1,\gamma_2,h}R=\min_{\gamma_1',\gamma_2',h'}R'.$$
\end{theorem}

Note that for conditionally independent observations, there is no need to distinguish between local and global sufficient statistics. We now establish Theorem 3 using the Bayesian risk for a two-sensor system

\begin{proof}
Let $$R_{\mathrm{min}}=\min_{\gamma_1,\gamma_2,h}R,$$ where the minimum Bayesian risk is achieved by the optimal quantizers $\gamma_i^*(\cdot)$ and estimator $h^*(\cdot)$. It is easy to see that Fig.~\ref{fig:decentralss2} can not achieve a better performance than $R_{\mathrm{min}}$ - for any given $T_i(\cdot)$ and $\gamma_i'(\cdot)$, one can simply define $\gamma_i(\Xbf_i)=\gamma'_i(T_i(\Xbf_i))$ whose performance is bounded by $R_{\mathrm{min}}$. Thus we only need to show that $R_{\mathrm{min}}$ can be achieved by Fig.~\ref{fig:decentralss2}, i.e., one can find $(\gamma_1'(\cdot),\gamma_2'(\cdot),h'(\cdot))$ that achieve $R_{\mathrm{min}}$ for the given sufficient statistics $T_1(\Xbf_1)$ and $T_2(\Xbf_2)$. Similar to the proof for the centralized case, it suffices to show that the optimal quantizers $\gamma_i^*(\mathbf{X}_i)$ achieving $R_\mathrm{min}$ depends on $\Xbf_i$ only through $T_i(\Xbf_i)$.

As $\mathbf{X}_1$ and $\mathbf{X}_2$ are conditionally independent,
\begin{align}
p(\mathbf{x}_1,\mathbf{x}_2,\theta)&=p(\theta)p(\mathbf{x}_1|\theta)p(\mathbf{x}_2|\theta)\nn\\
&=p(\mathbf{x}_1)p(\theta|\mathbf{x}_1)p(\mathbf{x}_2|\theta)\nn\\
&=p(\mathbf{x}_1)p(\theta|T_1(\mathbf{x}_1))p(\mathbf{x}_2|\theta).\nn
\end{align}
The last step comes from the fact that $T_1(\Xbf_1)$ is sufficent for the data $\Xbf_1$ and Lemma \ref{prob}. Expanding $R$ with respect to $\mathbf{X}_1$, we get
\begin{align}
R=&\int_{\theta}\int_{\mathbf{x}_1}\int_{\mathbf{x}_2}d[\theta,h(\gamma_1(\xbf_1),\gamma_2(\xbf_2))]p(\mathbf{x}_1,\mathbf{x}_2,\theta)d\mathbf{x}_2d\mathbf{x}_1 d\theta\nn\\
=&\int_{\theta}\int_{\mathbf{x}_1}\int_{\mathbf{x}_2}d[\theta,h(\gamma_1(\xbf_1),\gamma_2(\xbf_2))]p(\mathbf{x}_1)p(\theta|T_1(\mathbf{x}_1))\nn\\
 &\times p(\mathbf{x}_2|\theta)d\mathbf{x}_2d\mathbf{x}_1 d\theta\nn\\
\triangleq&§\int_{\mathbf{x}_1}\alpha_1(u_1,\mathbf{x}_1)p(\mathbf{x}_1)d\mathbf{x}_1,\nn
\end{align}
where
\bqa
\alpha_1(u_1,\mathbf{x}_1)\triangleq\int_{\mathbf{x}_2}\int_{\theta}d[\theta,h(u_1,u_2)]p(\theta|T_1(\mathbf{x}_1))p(\mathbf{x}_2|\theta)d\theta d\mathbf{x}_2.\nn
\eqa
Let $\gamma_2(\cdot)$ and $h(\cdot)$ take the form of the optimal $\gamma_2^*(\cdot)$ and $h^*(\cdot)$, $\gamma_1^*(\cdot)$ must be chosen such that the corresponding $\alpha_1(u_1,\mathbf{x}_1)$ is minimized. The condition for making $u_1=\gamma_1^*(\mathbf{x}_1)=s \in \{0,\dots,L-1\}$ given $\Xbf_1=\mathbf{x}_1$ is, for any $t\in\{0,\dots,L-1\}$,
\begin{align}
0\geq&~\alpha_1(s,\mathbf{x}_1)-\alpha_1(t,\mathbf{x}_1)\nn\\
=&~\int_{\mathbf{x}_2}\int_{\theta}\{d[\theta,h^*{(s,{\gamma_2^*(\mathbf{x}_2}))}]-d[\theta,h^*{(t,{\gamma_2^*(\mathbf{x}_2}))}]\}\nn\\
&\times p(\theta|T_1(\mathbf{x}_1))p(\mathbf{x}_2|\theta)d\theta d\mathbf{x}_2\nn,
%&=&\int_{\mathbf{x_2}\in \mathbf{R_0}}
%\int_{\theta}\{d[\theta,\theta_{00}]-d[\theta,\theta_{10}]\}p(\theta|T_1(\mathbf{x_1}))p(\mathbf{x_2}|\theta)d\theta d\mathbf{x_2}\nn\\
%&&-\int_{\mathbf{x_2}\in \mathbf{R_0}}
%\int_{\theta}\{d[\theta,\theta_{01}]-d[\theta,\theta_{11}]\}p(\theta|T_1(\mathbf{x_1}))p(\mathbf{x_2}|\theta)d\theta d\mathbf{x_2}\nn
\end{align}
which depends on $\Xbf_1$ only through $T_1(\mathbf{X}_1)$.

The optimal quantizer $\gamma_2^*(\cdot)$ at the second node, given that $\gamma_1(\cdot)$ and $h(\cdot)$ take the form of $\gamma_1^*(\cdot)$ and $h^*(\cdot)$, can be similarly shown to be a function of the sufficient statistic $T_2(\mathbf{X}_2)$. Thus we have established that both $\gamma_1^*(\cdot)$ and $\gamma_2^*(\cdot)$ can be equivalently expressed as functions of $T_1(\Xbf_1)$ and $T_2(\Xbf_2)$ respectively, i.e., there exist $\gamma_1'(\cdot)$ and $\gamma_2'(\cdot)$ such that
\begin{align}
\gamma_1'(T_1(\Xbf_1))&=\gamma_1^*(\Xbf_1),\\
\gamma_2'(T_2(\Xbf_2))&=\gamma_2^*(\Xbf_2).
\end{align}
Therefore, the above $\gamma_1'(\cdot)$ and $\gamma_2'(\cdot)$, together with $h'(\cdot)=h^*(\cdot)$, achieves $R_{\min}$ for Fig.~\ref{fig:decentralss}(b).
\end{proof}
The fact that likelihood ratio quantizer is optimal for decentralized detection with conditionally independent observations can be naturally derived from the above general result.
\begin{example}
Let $\theta \in \{0,1\}$ and its estimate $\hat{\theta} \in \{0,1\}$. The observations $\Xbf_1$ and $\Xbf_2$ are independent given $\theta$. Let $d(\cdot)$ take the form of $0-1$ cost, i.e., $d[\theta,\hat{\theta}]=0$ when $\theta=\hat{\theta}$ and $1$ otherwise. It is a trivial exercise to show $T_i(\xbf_i)={p(\xbf_i|\theta=1)\over{p(\xbf_i|\theta=0)}}$ is a sufficient statistic for $\theta$ with respect to $\Xbf_i$. Thus quantizing $T_i(\Xbf_i)$ is structurally optimal, which is consistent with \cite{Tsitsiklis:extremal} as the inference problem is precisely a hypothesis testing problem.
\end{example}
%\begin{figure}[tbh]
%\begin{center}
%\psfrag{x1}[c][c]{\scriptsize$\mathbf{X_1}$}
%\psfrag{x2}[c][c]{\scriptsize$\mathbf{X_2}$}
%\psfrag{T1}[c][c]{\scriptsize$T_1(\mathbf{X_1})$}
%\psfrag{T2}[c][c]{\scriptsize$T_2(\mathbf{X_2})$}
%\psfrag{u1}[c][c]{\scriptsize$U_1$}
%\psfrag{u2}[c][c]{\scriptsize$U_2$}
%\psfrag{p}[c][c]{\footnotesize$p(\mathbf{x}|\theta)$}
%\psfrag{s1}[c][c]{\footnotesize$T_1(\cdot)$}
%\psfrag{s2}[c][c]{\footnotesize$T_2(\cdot)$}
%\psfrag{g1}[c][c]{\footnotesize$\gamma_1(\cdot)$}
%\psfrag{g2}[c][c]{\footnotesize$\gamma_2(\cdot)$}
%\psfrag{h}[c][c]{\footnotesize$h(\cdot)$}
%\psfrag{t}[c][c]{\footnotesize$\theta$}
%\psfrag{t2}[c][c]{\footnotesize$\hat{\theta}$}
%\includegraphics[width=3.2in]{detralss21.eps}
%\caption[c]{\label{fig:decentralss1} An interactive detection network.}
%\end{center}
%\end{figure}

\subsection{Conditionally Dependent Observations}
While the previous section establishes the optimality of sufficiency based data reduction for conditionally independent observations even with quantization constraints, Example \ref{exp:counter suff} indicates that such is not the case with conditionally dependent observations. Nevertheless, in this section, we establish that within the problems involving dependent observations, there exist a class of problems such that quantizing sufficient statistics is still structurally optimal. Here we again utilize the HCI model \cite{Hchen:newframework}.
\begin{theorem}
\label{thm:cond de}
Let $\Wbf$ be a hidden variable such that (\ref{HCI}) is true. If $T_1({\Xbf_1})$ and $T_2({\Xbf_2})$ are local statistics that are sufficient with respect to $\mathbf{W}$, then quantizing $T_1(\Xbf_1)$ and $T_2(\Xbf_2)$ at the respective sensor is structurally optimal for the decentralized inference problem.
\end{theorem}

Note that the first Markov chain in (\ref{HCI}) indicates that $\Xbf_1$ and $\Xbf_2$ are conditionally independent given $\Wbf$. If $T_1({\mathbf{X}_1})$ and $T_2({\mathbf{X}_2})$ are locally sufficient for $\mathbf{W}$, $(T_1({\mathbf{X}_1}), T_2({\mathbf{X}_2}))$ is globally sufficient for $\Wbf$ and hence for $\theta$ by Corollary \ref{cor1}.

\begin{proof}
Let $R_{\mathrm{min}}$ be the minimum Bayesian risk achieved by Fig.~2(a) with the corresponding optimal quantizers $\gamma_i^*(\cdot)$, $i=1, 2$, and estimator $h^*(\cdot)$. We show that $\gamma_i^*(\mathbf{X}_i)$ is necessarily a function of the sufficient statistic $T_i(\mathbf{X}_i)$.

Without loss of generality, we assume that $\mathbf{W}$ is continuous. From (\ref{HCI}), we have
\begin{align}
\label{eqn:decentralss joint prob}
p(\mathbf{x}_1,\mathbf{x}_2|\theta)&=\int_{\mathbf{w}}p(\mathbf{x}_1,\mathbf{x}_2,\mathbf{w}|\theta)d\mathbf{w}\nn\\
&=\int_{\mathbf{w}}p(\mathbf{x}_1|\mathbf{w})p(\mathbf{x}_2|\mathbf{w})p(\mathbf{w}|\theta)d\mathbf{w}\nn\\
&=\int_{\mathbf{w}}{{p(\mathbf{w}|\mathbf{x}_1)p(\mathbf{x}_1)}\over{p(\mathbf{w})}}p(\mathbf{x}_2|\mathbf{w})p(\mathbf{w}|\theta)d\mathbf{w}.\nn\\
&=p(\mathbf{x}_1)\int_{\mathbf{w}}{{p(\mathbf{w}|T_1(\mathbf{x}_1))}\over{p(\mathbf{w})}}p(\mathbf{x}_2|\mathbf{w})p(\mathbf{w}|\theta)d\mathbf{w}.
\end{align}
Expanding $R$ with respect to $\mathbf{X}_1$, we obtain
\begin{align}
R=& \int_{\theta}\int_{\mathbf{x}_1}\int_{\mathbf{x}_2}d[\theta,h(\gamma_1(\xbf_1),\gamma_2(\xbf_2))]p(\mathbf{x}_1,\mathbf{x}_2,\theta)d\mathbf{x}_2d\mathbf{x}_1d\theta\nn\\
=&\int_{\theta}\int_{\mathbf{x}_1}\int_{\mathbf{x}_2}\int_{\mathbf{y}}d[\theta,h(\gamma_1(\xbf_1),\gamma_2(\xbf_2))]{{p(\mathbf{w}|T_1(\mathbf{x}_1))}\over{p(\mathbf{w})}}p(\mathbf{x}_1)\nn\\
&\times p(\mathbf{x}_2|\mathbf{w})p(\mathbf{w}|\theta)d\mathbf{w}d\mathbf{x}_2d\mathbf{x}_1d\theta\nn\\
\triangleq& \int_{\mathbf{x}_1}\alpha'_1(u_1,\mathbf{x}_1)p(\mathbf{x}_1)d\mathbf{x}_1,\nn
\end{align}
where
\begin{align}
\label{eqn:alphaprime}
\alpha'_1(u_1,\mathbf{x}_1)\triangleq&\int_{\theta}\int_{\mathbf{x}_2}\int_{\mathbf{w}}d[\theta,h(u_1,u_2)]{{p(\mathbf{w}|T(\mathbf{x}_1))}\over{p(\mathbf{w})}}p(\mathbf{x}_2|\mathbf{y})\nn\\&
\times p(\mathbf{w}|\theta)d\mathbf{w}d\mathbf{x}_2d\theta.
\end{align}
Therefore, given $\gamma_2^*(\cdot)$ and $h^*(\cdot)$, for $\gamma_1^*(\cdot)$ to achieve $R_\mathrm{min}$, $u_1=\gamma_1^*({\mathbf{x}_1})$ must be such that $\alpha'_1(u_1,\mathbf{x}_1)$ is minimized, i.e., $u_1=s \in \{0,\dots,L-1\}$ if
\bqa
i =  \arg\min\limits_{u_1}\alpha'_1(u_1,\mathbf{x}_1).\nn
\eqa
From (\ref{eqn:alphaprime}), $\gamma_1^*(\Xbf_1)$ depends on $\Xbf_1$ only through $T_1(\mathbf{X}_1)$. Similar argument shows that $\gamma_2^*(\cdot)$ is also a function of the sufficient statistic $T_2(\mathbf{X}_2)$.
\end{proof}

The key to applying the above result also depends largely on a well chosen $\Wbf$ for the HCI model. As discussed in Section \ref{SC:exp}, the choice of $\Wbf$ can often be obtained by careful examination of the signal model. We now continue with Examples 2 and 3 by adding quantization constraints to the respective problems.
%For example, the na\"{\i}ve choice of $\Wbf=(\Xbf_1,\Xbf_2)$, while satisfying the defining Markov chains, does not result in any data reduction as the sufficient statistics for the data are nothing but the original data. 

\begin{example}
Consider Example \ref{exp:globalsuff} under the quantization constraint, i.e., we need to estimate $\theta$ based on the quantized version of $\Xbf_1$ and $\Xbf_2$. Since $\sum_j{X_{1j}}$ and $\sum_j{X_{2j}}$ are locally sufficient for the hidden variable $W$, quantizing $\sum_j{X_{1j}}$ and $\sum_j{X_{2j}}$ is structurally optimal by Theorem \ref{thm:cond de}.
\end{example}
\begin{example}
Consider Example \ref{eg:cooperative} when quantization is needed at each node. In Example \ref{eg:cooperative} we have shown that $\{|X_k|\}, k=1,\cdots,K$ are globally sufficient for $H$. Therefore, from Theorem \ref{thm:cond de}, quantizing $|X_k|$ at the $k$th sensor is structurally optimal. This result is consistent with that in \cite{Peng:energydetection} which shows that the optimal detector at each local sensor is an energy detector for the corresponding cooperative spectrum sensing problem, i.e., in the form of a threshold test using $|X_k|^2$.
\end{example}

As a final technical note, the conditional independent case can be considered as a special case of Theorem 4. Specifically, setting $\Wbf=\theta$, one can see that Theorem 3 follows naturally from Theorem 4.
\subsection{An Alternative Condition for Structural Optimality}
\label{SC:generalcond}
Theorems 3 and 4 established the structural optimality of sufficiency based data reduction with independent data and with dependent data under a given HCI structure, respectively. In this section, we provide an alternative characterization that encompasses both cases. To proceed, we note that in Theorems 3 and 4 the joint distribution $p(\xbf_1,\xbf_2,\theta)$ can be expressed in both cases as the product of $p(\xbf_1)$ and a nonnegative function of $T_1(\xbf_1)$, $\xbf_2$ and $\theta$. We show that this factorization is indeed what is needed to establish that quantizing $T_1(\Xbf_1)$ achieves the same optimal inference performance as quantizing $\Xbf_1$ given that the optimal quantizer $\gamma_2^*(\cdot)$ and the optimal estimator $h^*(\cdot)$ are used at the second sensor and at the fusion center respectively.
\begin{theorem}
\label{thm:general condition}
If there exist two nonnegative functions $g(\cdot)$ and $h(\cdot)$ and a statistic $T_1(\Xbf_1)$ such that
\bqa
\label{eqn:generalcondi}
p(\xbf_1,\xbf_2,\theta)=g(\xbf_1)f(T_1(\xbf_1),\xbf_2,\theta),
\eqa
then quantizing $T_1(\Xbf_1)$ achieves the same optimal inference performance as quantizing $\Xbf_1$.
\end{theorem}

From (\ref{eqn:generalcondi}), if we marginalize $\Xbf_2$ on both sides, we have
\bqa
p(\xbf_1,\theta)=g(\xbf_1)\int_{\xbf_2}f(T_1(\xbf_1),\xbf_2,\theta)d\xbf_2.\nn
\eqa
Thus, by the factorization theorem \cite{Casella:stata}, (\ref{eqn:generalcondi}) implies that $T_1(\Xbf_1)$ is a local sufficient statistic for $\theta$.

\begin{proof}
Let $R_\mathrm{min}$ be the minimum Bayesian risk achieved by Fig.~\ref{fig:decentralss1} with quantizer $\gamma_i^*(\cdot)$ and estimator $h^*(\cdot)$. We show that, if (\ref{eqn:generalcondi}) holds, then $\gamma_1^*(\Xbf_1)$ depends on $\Xbf_1$ only through the sufficient statistic $T_1(\Xbf_1)$.

Again, expanding $R$ with respect to $\Xbf_1$, we get
\begin{align}
R=&\int_{\theta}\int_{\xbf_1}\int_{\mathbf{x}_2}d[\theta,h(\gamma_1(\xbf_1),\gamma_2(\xbf_2)]p(\mathbf{x}_1,\mathbf{x}_2,\theta)d\mathbf{x}_2d\mathbf{x}_1d\theta\nn\\
=&\int_{\theta}\int_{\xbf_1}\int_{\mathbf{x}_2}d[\theta,h(\gamma_1(\xbf_1),\gamma_2(\xbf_2)]g(\mathbf{x}_1)f(T_1(\xbf_1),\xbf_2,\theta)\nn\\&\times d\mathbf{x}_2d\mathbf{x}_1d\theta\nn\\
\triangleq& \int_{\mathbf{x}_1}\alpha_1''(u_1,\mathbf{x}_1)g(\mathbf{x}_1)d\mathbf{x}_1,\nn
\end{align}
where
\begin{align}
\alpha_1''(u_1,\mathbf{x}_1)\triangleq&\int_{\theta}\int_{\mathbf{x}_2}d[\theta,h(u_1,\gamma_2(\xbf_2))]f(T_1(\xbf_1),\xbf_2,\theta)\nn\\&\times d\mathbf{x}_2d\theta.\nn
\end{align}
Given the optimal second quantizer $\gamma_2^*(\cdot)$ and estimator $h^*(\cdot)$, $\gamma_1^*(\cdot)$ must be such that it minimizes $\alpha_1''(u_1,\xbf_2)$, i.e., $u_1=\gamma_1^*(\xbf_1)=s \in \{0,\dots,L-1\}$ if for any $t\in \{0,\dots,L-1\}$,
\bqa
0\geq\alpha_1''(s,\xbf_1)-\alpha_1''(t,\xbf_1).\nn
\eqa
The proof is thus complete by recognizing that $\alpha_1''(u_1,\xbf_1)$ depends on $\xbf_1$ only through $T(\xbf_1)$.
\end{proof}

Theorem \ref{thm:general condition} provides an alternative way of formulating the sufficiency based data reduction, i.e., one may directly check the joint probability $p(\xbf_1,\xbf_2,\theta)$ instead of searching for a meaningful hidden variable $\Wbf$. While Theorem 5 appears to be more general than Theorem 4, we show in the following that these two theorems are indeed equivalent to each other in that they imply each other. This observation is also consistent with the fact that the conditional independence case can also be considered as a sepcial case of the HCI model.

\begin{proposition}
Thereoms 4 and 5 are equivalent.
\end{proposition}
\begin{proof}
The direction that Theorem 5 implies Theorem 4 is trivial as (\ref{eqn:decentralss joint prob}) satisfies (\ref{eqn:generalcondi}). We now show the other direction. Notice that given (\ref{eqn:generalcondi}), $\theta-(T_1(\Xbf_1),\Xbf_2)-(\Xbf_1,\Xbf_2)$ form a Markov chain and 
\begin{align}
p(\xbf_1,\xbf_2)=&~\int_\theta p(\xbf_1,\xbf_2,\theta)d\theta\nn\\
=&~\int_\theta g(\xbf_1)f(T_1(\xbf_1),\xbf_2, \theta) d\theta\nn\\
=&~g(\xbf_1)\int_\theta f(T_1(\xbf_1),\xbf_2, \theta) d\theta\nn,
\end{align}
which shows that $\Xbf_1-T_1(\Xbf_1)-\Xbf_2$ and hence $\Xbf_1-(T_1(\Xbf_1),\Xbf_2)-\Xbf_2$ are two Markov chains. Combining $\theta-(T_1(\Xbf_1),\Xbf_2)-(\Xbf_1,\Xbf_2)$ and $\Xbf_1-(T_1(\Xbf_1),\Xbf_2)-\Xbf_2$, we can choose $(T_1(\Xbf_1),\Xbf_2)$ as our hidden variable in the HCI model. That $\Xbf_1-T_1(\Xbf_1)-\Xbf_2$ is a Markov chain also implies that $(T_1(\Xbf_1),\Xbf_2)-T_1(\Xbf_1)-\Xbf_1$ form a Markov chain. Then $T_1(\Xbf_1)$ is a sufficient statistic for $(T_1(\Xbf_1),\Xbf_2)$ with respect to $\Xbf_1$ and achieves the structural optimality by Theorem 4.
\end{proof}

The fact that (\ref{eqn:generalcondi}) implies that $T_1(\Xbf_1)$ is a sufficient statistic for $\Xbf_1$ does not mean $T_1(\Xbf_1)$ being a sufficient statistic is a necessary condition for optimality. This is because (\ref{eqn:generalcondi}) itself is only a sufficient condition for optimality. Given below is a trivial example illustrating that a local statistic which achieves optimality is not necessarily a sufficient statistic.
\begin{example}
For $i=1,\cdots,n$, let
\bqa
X_{1i}&=&\theta+W_i,\nn\\
X_{2i}&=&\theta+V_i,\nn
\eqa
where $\theta, W_1, \cdots, W_n, V_1, \cdots, V_n$ are mutually independent Gaussian random variables such that $\theta\sim \mathcal{N}(0,1)$, $W_j\sim \mathcal{N}(0,1)$, $V_j\sim \mathcal{N}(0,1)$. Then $\Xbf_1$ and $\Xbf_2$ are conditionally independent given $\theta$. It is also clear that $\sum_i{X_{1i}}$ and $\sum_i{X_{2i}}$ are locally sufficient for $\theta$, thus quantizing $\sum_i{X_{1i}}$ and $\sum_i{X_{2i}}$ can achieve the optimal inference with corresponding quantizers $\gamma_1^*(\cdot)$ and $\gamma_2^*(\cdot)$ and the optimal estimator $h^*(\cdot)$.

Now consider another local statistic $U(\Xbf_1)=\gamma_1^*(\sum_i{X_{1i}})\in\{0,1\}$. If we quantize this statistic instead of $\sum_i{X_{1i}}$ at the first node while using $\gamma_2^*(\cdot)$ at the second node and $h^*(\cdot)$ at the fusion center, the optimal inference is also guaranteed, although the corresponding quantize is for $U(\Xbf_1)$ is a degenerate one, i.e., an identity mapping. It is trivial to see that $U(\Xbf_1)$ is not a sufficient statistic for $\theta$.
%Let $\Xbf_i=\{X_{i1},\cdots,X_{in}\}$ with $X_{ij}\sim \mathcal{N}(\theta,1)$ and $X_{ij}$'s are mutually independent for $i=1,2, j=1,\cdots,n$ in Fig.
\end{example}

\section{Conclusion and Discussions}
\label{conclusion}
In this paper we have extended the sufficiency principle for decentralized data reduction with dependent observations where data reduction needs to be done locally at distributed sensors. For the conditional independence case, local sufficiency and global sufficiency imply each other. For the dependent observations, however, there is no definitive connection between the two notions of sufficiency in general. Using a recently proposed HCI model, we establish conditions under which local sufficiency implies global sufficiency with dependent observations.

A more interesting, and practically more important question is the study of decentralized data reduction when each sensor is subject to a quantization constraint. We do not address explicit quantizer design in this work; instead, we find sufficient conditions such that a separation approach, namely data reduction followed by a quantizer, is structurally optimal under the Bayesian inference framework for both centralized inference and decentralized inference with conditionally independent observations. For decentralized inference with conditionally dependent observations, quantizing sufficient statistics, even global ones, need not be optimal. Nevertheless, utilizing the HCI model, we have provided a suitable way of finding optimal data reduction if it exists. %We have also established a unifying condition that encompasses both the independent and the dependent observation cases.

%Let us reconsider Example \ref{exp:counter suff} where $\Xbf_1=\Xbf_2$. Now that quantizing $(\Xbf_1,\varnothing)$, which is globally sufficient, does not achieve the optimal inference, one might ask that what local statistics can be used to achieve the same optimal inference performance as the raw data. Since this problem is equivalent to a centralized inference problem with a 2-bit quantizer, Theorem 1 implies that quantizing the minimal sufficient statistic $M(\Xbf_1)$ at each sensor can achieve the optimal inference. But a minimal sufficient statistic is a function of any other sufficient statistic \cite{Casella:stata}, thus any local sufficient statistics $(T_1(\Xbf_1),T_2(\Xbf_2))$ at each node attains the structural optimality. 

While Theorem 4 helps identify cases where meaningful data reduction can be achieved for dependent observations, identifying suitable hidden variable $\Wbf$ often requires careful examination and a good insight into the signal model. An alternative, yet equivalent formulation was provided in Theorem 5 which will be useful when a closed-form likelihood function of all data can be obtained. 

There are still cases where the existing tools developed in the present paper are not sufficient. We use the degenerate signal model in Example 4 to illustrate this point. Recall that for the case with $\Xbf_1=\Xbf_2$ and $1$-bit quantizer at each node, the optimal decentralized quantization is equivalent to a $2$-bit quantization of the observation in a centralized inference system. Clearly, the optimum data reduction would be to find the minimum sufficient statistic [2] prior to quantization. As the minimum sufficient statistic is a function of any other sufficient statistic, it is apparent that any locally sufficient statistic pair $(T_1(\Xbf_1),T_2(\Xbf_2))$ retains the optimal inference performance.

The above argument can also be made more rigorous by expanding the Bayesian risk. Let $T_1(\Xbf_1)$ be a sufficient statistic for $\theta$ with respect to $\Xbf_1$. When $\Xbf_1=\Xbf_2$,
\begin{align}
\label{eqn: example thm6}
p(\xbf_1,\xbf_2,\theta)=&~p(\xbf_1,\theta)\delta(\xbf_1-\xbf_2)\nn\\
\stackrel{(a)}{=}&~p(\xbf_1)p(\theta|T_1(\xbf_1))\delta(\xbf_1-\xbf_2),
\end{align}
where $\delta(\cdot)$ is the Dirac delta function and $(a)$ is from Lemma \ref{prob}. Then we have
\begin{align}
R=&~\int_{\theta}\int_{\xbf_1}\int_{\xbf_2}d[\theta,h(\gamma_1(\xbf_1),\gamma_2(\xbf_2))]p(\xbf_1,\xbf_2,\theta)d\xbf_1d\xbf_2d\theta\nn\\
=&~\int_{\theta}\int_{\xbf_1}\int_{\xbf_2}d[\theta,h(\gamma_1(\xbf_1),\gamma_2(\xbf_2))]p(\xbf_1)p(\theta|T_1(\xbf_1))\nn\\
&~\times \delta(\xbf_1-\xbf_2)d\xbf_1d\xbf_2d\theta\nn\\
=&~\int_{\theta}\int_{\xbf_1}d[\theta,h(\gamma_1(\xbf_1),\gamma_2(\xbf_1))]p(\xbf_1)p(\theta|T_1(\xbf_1))d\xbf_1d\theta\nn\\
\triangleq&\int_{\xbf_1}\alpha_1(u_1,\xbf_1)p(\xbf_1)d\xbf_1\nn,
\end{align}
where 
\begin{align}
u_1\triangleq&~\gamma_1(\xbf_1),\nn\\
\alpha_1(u_1,\xbf_1)\triangleq&\int_{\theta}d[\theta,h(\gamma_1(\xbf_1),\gamma_2(\xbf_1))]p(\theta|T_1(\xbf_1))d\theta\nn.
\end{align}
Given the optimal quantizer $\gamma_2^*(\cdot)$ and estimator $h^*(\cdot)$, we see that $\alpha_1(u_1,\xbf_1)$ depends on $\xbf_1$ only through $T_1(\xbf_1)$. The same argument shows that $\gamma_2^*(\cdot)$ is a function of the sufficient statistic $T_2(\Xbf_2)$ at the second node. Thus any local sufficient statistics $(T_1(\Xbf_1),T_2(\Xbf_2))$ can be used to achieve the same optimal inference performance. Note that while any local sufficient statistics $(T_1(\Xbf_1),T_2(\Xbf_2))$ preserve the optimal inference performance for this degraded observation model, they may not achieve the same degree of data reduction as that of the minimal sufficient statistic.

However, it is clear that Theorems 4 and 5 do not apply to this example, as the joint probability (\ref{eqn: example thm6}) can not be formulated in the form of  (\ref{eqn:generalcondi}). Searching for more general conditions to ensure the structural optimality of data reduction in the presence of quantization constraint will be our future work.

\section*{Acknowledgment}
The authors would like to thank Pengfei Yang for many helpful discussions and in particular for helping establish the equivalence of Theorem 4 and Theorem 5. This work is supported in part by National Science Foundation under Award CCF1218289, by Army Research Office under Award W911NF-12-1-0383, and by Air Force Office of Scientific Research under Award FA9550-10-1-0458. 
%Note that while the results in this paper were derived  for the case that the parameter of interest is random, they extend naturally to that of unknown deterministic parameters. Indeed parallel derivations can be constructed using the Neyman-Factorization theorem instead of the Markov chain condition used in the present paper.

\bibliographystyle{IEEEbib}

\end{document}